\documentclass[a4paper,onecolumn,11pt,accepted=2025-02-11]{quantumarticle}
\pdfoutput=1
\usepackage[utf8]{inputenc}
\usepackage[english]{babel}
\usepackage[T1]{fontenc}
\usepackage{amsmath}
\usepackage{hyperref}

\usepackage{tikz}
\usepackage{lipsum}

\usepackage{graphicx}
\usepackage{amsmath}
\usepackage{latexsym}
\usepackage{bbm}
\usepackage{enumerate,appendix}
\usepackage{amsmath, amsthm, amssymb,commath,mathtools}
\usepackage{color,calc,graphicx}

\usepackage{acronym}
\acrodef{RAC}{random access code}
\acrodef{ASP}{average success probability}
\acrodef{POVM}{positive operator-valued measure}
\acrodef{PVM}{projection-valued measure}
\acrodef{MUB}{mutually unbiased base}
\acrodef{SDP}{semidefinite programming}

\newcommand{\ket}[1]{\vert#1\rangle}

\newcommand{\braket}[2]{\langle#1\vert#2\rangle}

\newcommand{\ketbra}[2]{\vert#1\rangle\langle#2\vert}

\newtheorem{result}{Result}

\newtheorem{lemma}{Lemma}
\newtheorem{corollary}{Corollary}

\usepackage{physics}

\renewcommand{\Tr}[1]{\mathrm{tr}\left[#1\right]}

\usepackage{booktabs}
\usepackage{multirow}
\usepackage{subfigure}
\usepackage{dcolumn}
\usepackage{mathrsfs}
\usepackage{diagbox}
\usepackage{array}
\usepackage{makecell}
\usepackage{threeparttable}

\begin{document}

\title{Simple and general bounds on quantum random access codes}

\author{M\'at\'e Farkas}
\thanks{mate.farkas@york.ac.uk}
\affiliation{Department of Mathematics, University of York, Heslington, York, YO10 5DD, United Kingdom}

\author{Nikolai Miklin}
\thanks{nikolai.miklin@tuhh.de}
\affiliation{Institute for Quantum-Inspired and Quantum Optimization, Hamburg University of Technology, Germany}

\author{Armin Tavakoli}
\thanks{armin.tavakoli@teorfys.lu.se}
\affiliation{Physics Department and NanoLund, Lund University, Box 118, 22100 Lund, Sweden.}

\maketitle

\begin{abstract}
Random access codes are a type of communication task that is widely used in quantum information science. The optimal average success probability that can be achieved through classical strategies is known for any random access code. However, only a few cases are solved exactly for quantum random access codes. In this paper, we provide bounds for the fully general setting of $n$ independent variables, each selected from a $d$-dimensional classical alphabet and encoded in a $D$-dimensional quantum system subject to an arbitrary quantum measurement. The bound recovers the exactly known special cases, and we demonstrate numerically that even though the bound is not tight overall, it can still yield a good approximation. 
\end{abstract}

\section{Introduction}
Quantum \acp{RAC} are a broadly useful tool in quantum information science. In addition to being studied on their own merit (see e.g.~\cite{Nayak1999, Hayashi2006, Pawlowski2010, Tavakoli2016, Liabotro2017, Doriguello2021, Miao2022, Silva2023}), an incomplete list of their broader relevance includes protocols for quantum contextuality \cite{Spekkens2009}, information-theoretic principles for quantum correlations \cite{Pawlowski2009}, tests of quantum dimension \cite{Brunner2013, Pauwels2021}, quantum cryptography \cite{Pawlowski2011}, famous open problems in Hilbert space geometry \cite{Aguilar2018} and certification of measurements \cite{Mironowicz2019, Smania2020, Carmeli2020} and instruments \cite{Mohan2019, Miklin2020}. This widespread use has led to quantum RACs being the focus of many experiments, see e.g.~\cite{Spekkens2009, Muhammad2014, Tavakoli2015, Aguilar2018, Foletto2020, Anwer2020, Xiao2021}. To prove and maximize the utility of RACs in most of these tasks, it is essential to find optimal quantum RAC strategies, or at least to find relatively tight bounds on the optimal performance. This is because  a tight upper bound is necessary e.g.~in order to use quantum RACs for certification \cite{Tavakoli2018,Farkas2019}, whereas approximate bounds can lead to applications in e.g.~quantum key distribution \cite{Pawlowski2011,WP15}. Finding such universal bounds is precisely the aim of this work.

Consider a communication scenario in which a sender encodes private data into a message that is sent to a receiver who wants to recover some freely chosen part of the original data set. \acp{RAC} are a particularly natural class of such tasks. In a \ac{RAC}, the private data can consist of $n$ independent and uniformly distributed classical variables, $x\coloneqq \{x_1,x_2,\ldots,x_n\}$. Each variable is selected from an alphabet with $d$ distinct symbols, $x_i\in[d\,]\coloneqq \{1,2,\ldots,d\}$ for $i=1,2,\ldots,n$. The data set $x$ is then encoded by the sender, Alice, into a typically much smaller message whose dimension is $D$. The message is sent to the receiver, Bob, who privately selects at random which element in the data set he wishes to recover, labeled by $y\in[n]$, and outputs $b\in[d\,]$ as his guess for the value of the random variable $x_y$. The \ac{ASP} of the \ac{RAC} is hence given by
\begin{equation}
\mathcal{P}_{n,d,D}\coloneqq \frac{1}{nd^n}\sum_{x\in[d\,]^n}\sum_{y=1}^n \mathbb{P}[b=x_y|x,y],
\end{equation}
where, following a common short-hand notation, we used $x,y$, and $x_y$ as labels of the corresponding random variables' outcomes.

In a classical \ac{RAC}, the message is represented by a random variable with $D$ possible integer values. The encoding is then represented by any function $E:[d\,]^n\rightarrow [D\,]$ and, similarly, the decoding consists of a set of functions $D_y:[D\,]\rightarrow [d\,]$. The case in which the message dimension equals the size of the alphabet, i.e.,~when $D=d$, has been the most studied case so far. The optimal \ac{ASP} for such protocols was conjectured in \cite{Tavakoli2015} and later proven in \cite{ambainis2015optimal}. In the general case, namely when $D\neq d$, the optimal \ac{ASP} was recently proven in \cite{Saha2023}. 

In a quantum \ac{RAC}, the sender encodes the classical data $x$ into a quantum state $\rho_x$ whose Hilbert space dimension is $D$. The receiver's decoding corresponds to a set of quantum measurements $(M_{b|y})_{b\in[d\,]}$, for $y\in[n]$, where $b$ denotes the outcome and $y$ denotes the setting. For each $y$, any $D$-dimensional \ac{POVM} is a valid decoding operation. Hence, in the quantum \ac{RAC}, the optimal \ac{ASP} becomes
\begin{equation}\label{ASP}
\mathcal{P}^Q_{n,d,D} \coloneqq \max_{\rho_x,M_{b|y}}   \frac{1}{nd^n}\sum_{x\in[d\,]^n}\sum_{y=1}^n \Tr{\rho_x M_{x_y|y}}.
\end{equation}

Finding the optimal quantum \ac{ASP} is, in general, not easy. However, one family of exact results is known. This pertains to the case of $n=2$ and $D=d$. It was conjectured in \cite{Tavakoli2015} and later proven in \cite{Farkas2019} that 
\begin{equation}\label{n2}
	\mathcal{P}^Q_{2,d,d}=\frac{1}{2}\left(1+\frac{1}{\sqrt{d}}\right).
\end{equation}
This can be achieved by selecting the two decoding measurements as the computational and Fourier bases measurements. The encodings are obtained via the Weyl-Heisenberg group generators as $X^{x_1}Z^{x_2}\ket{\psi}$, where $X=\sum_{k=0}^{d-1} \ketbra{k+1}{k}$ and $Z=\sum_{k=0}^{d-1}e^{\frac{2\pi \mathrm{i}}{d}k}\ketbra{k}{k}$, with $\ket{\psi}$ being a uniform superposition of $\ket{0}$ and its Fourier transform \cite{Tavakoli2015}. This protocol is also unique, in a weaker self-testing sense: any pair of \acp{MUB} leads to an optimal strategy, and not all of these are unitarily equivalent to the computational and Fourier bases~\cite{Farkas2019}.

One more exact result is known, namely when Alice has three bits and communicates a qubit to Bob. In that case, one has
\begin{equation}\label{case2}
	\mathcal{P}^Q_{3,2,2}=\frac{1}{2}\left(1+\frac{1}{\sqrt{3}}\right).
\end{equation}
The optimal  quantum protocol consists in measuring the three Pauli observables and preparing eight qubit states on the Bloch sphere so that they form a cube \cite{Ambainis1999, Ambainis2002}. The protocol is  known to be unique in a self-testing sense~\cite{Tavakoli2018}.

Beyond these exact results, a generic bound on the optimal quantum ASP when $d=D=2$ is known~\cite{ambainis2009quantum} to be  
\begin{equation}\label{d2}
	\mathcal{P}^Q_{n,2,2} \leq \frac{1}{2}\left(1+\frac{1}{\sqrt{n}}\right).
\end{equation} 
This was later generalized in~\cite{Vicente2019} to a bound valid for arbitrary tuples $(n,d,D)$, which reads
\begin{equation}\label{julio}
	\mathcal{P}^Q_{n,d,D}\leq \frac{1}{d}+\frac{\sqrt{d D}-1}{d \sqrt{n}}.
\end{equation}
Eq.~\eqref{julio} reduces to Eq.~\eqref{d2} when $d=D=2$ and therefore subsumes that bound. Note that it does not recover the exact family of results in Eq.~\eqref{n2}. Other known bounds for quantum RACs rely mainly on relaxation methods based on  \ac{SDP} \cite{tavakoli2023semidefinite}. These are commonly applied to specific and typically small-scale tuples $(n,d,D)$, see e.g.,~\cite{Navascues2015, Rosset2019, Pauwels2022} for methods and case studies. It is apparent from these studies that \ac{SDP} techniques scale rather badly with the problem size. It is therefore highly desirable to develop new analytic tools for bounding the quantum \ac{ASP} for generic tuples $(n,d,D)$, which is precisely the aim of our work.

In this paper, we present two simple analytical upper bounds on the quantum \ac{ASP} which apply to the most general setting, namely any tuple $(n,d,D)$. Our final, combined, bound is given in Corollary~\ref{Cor1}. It is distinct from the known general bound in Eq.~\eqref{julio}. Notably, it recovers the both exact results in Eq.~\eqref{n2} and Eq.~\eqref{case2} as special cases, as well as the bound in Eq.~\eqref{d2}. Since our bound is typically not tight, we study its performance through numerical case studies. We demonstrate with examples that it can be a good approximation of the exact value.

\section{Bounds on quantum average success probability}

Operator norm inequalities have proven to be a useful analytic tool for bounding the quantum \ac{ASP} of \acp{RAC} \cite{Farkas2019}. In order to extend these techniques beyond two measurement settings, we employ the following lemma, which turns out to be highly useful in various ways for obtaining improved analytic bounds for more than two measurement settings.
\begin{lemma}\label{Lemma}
	Let $A$ be a trace-zero Hermitian matrix. Then it holds that
	\begin{equation}\label{lemma}
	\norm{A}_\infty \leq\sqrt{\frac{r-1}{r}} \norm{A}_F\ ,
	\end{equation}
	where $r$ is the rank of $A$, $\norm{\cdot}_\infty$ and $\norm{\cdot}_F$ denote the operator and the Frobenius norms, respectively. 
\end{lemma}

\begin{proof}
Matrix $A$ has $r$ real eigenvalues which we denote as $\{\lambda_i\}_{i=1}^r$. Since $\Tr{A}=0$, we have that $\sum_{i=1}^r \lambda_i=0$. Let $\lambda_1$ be the largest eigenvalue in the absolute value, which implies that $\norm{A}_\infty =\abs{\lambda_1}$. The Frobenius norm can then be lower-bounded as
\begin{equation}
\begin{split}
\norm{A}_F^2 & \left. =\sum_{i=1}^r \lambda_i^2\geq \lambda_1^2 +\frac{1}{r-1} \left(\sum_{i=2}^r \abs{\lambda_i}\right)^2 \geq \lambda_1^2 +\frac{1}{r-1} \left(\sum_{i=2}^r \lambda_i\right)^2 = \lambda_1^2 +\frac{1}{r-1} \left(-\lambda_1\right)^2 \right. \\
& \left. =\frac{r}{r-1}\norm{A}_\infty^2.
\right.
\end{split}
\end{equation}
In the first inequality we used the well-known relation $\norm{\cdot}_F\geq \frac{1}{\sqrt{t}} \norm{\cdot}_1$, where $ \norm{\cdot}_1$ is the trace norm and $t$ is the dimension of the relevant operator. In the second inequality we discarded the absolute values, and in the next step we used the trace-zero condition.  Re-arranging the left- and right-hand side returns Eq.~\eqref{lemma}.
\end{proof}

We now state and prove our first bound on the quantum \ac{ASP}. 
\begin{result}\label{Res1}
	The average success probability of the quantum random access code, in the setting of $n$-element data set with alphabet size $d$ and message dimension $D$, is bounded as
	\begin{equation}\label{res1}
\mathcal{P}^Q_{n,d,D} \leq \frac{1}{d}+\frac{D-1}{\sqrt{ndD}}\ .
	\end{equation}	
\end{result}
\begin{proof}
We can trivially re-write the quantum \ac{ASP} in Eq.~\eqref{ASP} as
\begin{align}\label{step}
\mathcal{P}^Q_{n,d,D}=  \frac{1}{nd^n}\sum_x \sum_y\frac{1}{D}\Tr{M_{x_y|y}}+\frac{1}{nd^n}\sum_{x} \Tr{\rho_x \sum_y \left(M_{x_y|y}-\frac{\openone}{D}\Tr{M_{x_y|y}}\right)}\ , 
\end{align}
where we omitted writing the maximization and limits of the sums for convenience. In Eq.~\eqref{step}, $\openone$ is the identity operator on the $D$-dimensional Hilbert space in which $\rho_x$ and $M_{b|y}$ are defined.
 From the normalization, $\sum_{b=1}^{d}M_{b|y}=\openone$, it follows that the first term evaluates to $\frac{1}{d}$. Clearly, the optimal choice of $\rho_x$ corresponds to the eigenvector with the largest eigenvalue of the Hermitian operator $O_x\coloneqq \sum_y \left(M_{x_y|y}-\frac{\openone}{D}\Tr{M_{x_y|y}}\right)$. This gives
\begin{equation}
\mathcal{P}^Q_{n,d,D}=\frac{1}{d}+\frac{1}{nd^n}\sum_{x} \norm{O_x}_\infty\ .
\end{equation}
Notice that by adopting the form~\eqref{step}, we have conveniently ensured that $\Tr{O_x}=0$. Hence, we can now apply Lemma~\ref{Lemma} to obtain 
\begin{equation}\label{step3}
\mathcal{P}^Q_{n,d,D} \leq  \frac{1}{d}+ \frac{1}{n d^n}\sqrt{\frac{D-1}{D}}\sum_x \norm{O_x}_F\leq \frac{1}{d}+ \frac{\sqrt{D-1}}{n \sqrt{Dd^{n}}} \sqrt{\sum_x \Tr{O_x^2}}\ ,
\end{equation}
where in the second step we used the concavity inequality $\frac{1}{N} \sum_{i=1}^N \sqrt{t_i}\leq \sqrt{\frac{1}{N}\sum_{i=1}^N t_i}$, which holds for any $t_i\geq 0$, $i\in [N]$. Simplifying the expression under the square-root gives
\begin{equation}\label{step2}
\sum_x \Tr{O_x^2}=\sum_x \sum_{y,z=1}^n \Tr{M_{x_y|y}M_{x_z|z}}-\frac{1}{D}\sum_x \sum_{y,z=1}^n \Tr{M_{x_y|y}}\Tr{M_{x_z|z}}.
\end{equation}
The two terms on the right-hand-side of Eq.~\eqref{step2} simplify, respectively, as
\begin{align}\begin{split}
\sum_x \sum_{y,z=1}^n \Tr{M_{x_y|y}M_{x_z|z}}&=\sum_{y}\sum_x \Tr{M_{x_y|y}^2} +\sum_{y\neq z} \sum_{x\setminus\{x_y,x_z\}}\sum_{x_y,x_z}\Tr{M_{x_y|y}M_{x_z|z}} \\
&= \sum_{y}\sum_x \Tr{M_{x_y|y}^2} +n(n-1)d^{n-2}D\ ,
\end{split}\end{align}
and
\begin{align}\begin{split}
\frac{1}{D}\sum_x \sum_{y,z=1}^n \Tr{M_{x_y|y}}\Tr{M_{x_z|z}}&=\frac{1}{D}\sum_{y}\sum_{x} \Tr{M_{x_y|y}}^2 \\
& +\frac{1}{D} \sum_{y\neq z}\sum_{x\setminus \{x_y,x_z\}}\sum_{x_y} \Tr{M_{x_y|y}}\sum_{x_z}\Tr{M_{x_z|z}}\\
&\geq \frac{1}{D}\sum_{y}\sum_{x} \Tr{M_{x_y|y}^2}+n(n-1)d^{n-2}D\ ,
\end{split}\end{align}
where in the last inequality we used that for $M\geq 0$, we have $\Tr{M^2} \leq \Tr{M}^2$. 
In the above two equations, $\sum_{y \neq z}$ denotes the summation over $y\in [n]$ and $z\in [n]$ such that $y \neq z$.
Substituting this back into Eq.~\eqref{step2}, we obtain
\begin{equation}\label{step2_2}
\sum_x \Tr{O_x^2}\leq  \frac{D-1}{D}\sum_y \sum_{x}\Tr{M_{x_y|y}^2}\leq \frac{D-1}{D}\sum_y \sum_{x\setminus\{x_y\}} \sum_{x_y}\Tr{M_{x_y|y}}=(D-1)nd^{n-1}\ ,
\end{equation}
where we used that for an operator $0 \le M \le \openone$, we have $\Tr{M^2} \leq \Tr{M}$. Finally, substituting this back into Eq.~\eqref{step3} returns the result in Eq.~\eqref{res1}.
\end{proof}
Notice that when $d=D=2$, we recover the known bound \eqref{d2}. Also, when $d=D$, we recover the previously known bound in Eq.~\eqref{julio}. However, for $d=D$ and $n=2$ the above upper bound coincides with the other known bound \eqref{n2} only for $d=2$, and otherwise provides a weaker bound. 

We now state and prove the second main result, which is a different general bound on the quantum \ac{ASP}. This time, it is more relevant for the scaling in $d$, as it recovers the bound in Eq.~\eqref{n2} as a special case. In this sense, it is  complementary to Result~\ref{Res1}.

\begin{result}\label{Res2}
The average success probability of the quantum random access code, in the setting of $n$-element data set with alphabet size $d$ and message dimension $D$, is bounded as
	\begin{equation}\label{res2}
	\mathcal{P}^Q_{n,d,D} \leq \frac{1}{n}\left(1+(n-1)\frac{\sqrt{D}}{d}\right) .
	\end{equation}	
\end{result}

\begin{proof}
Our main tool is an operator inequality proved by Popovici and Sebesty\'en \cite{Popovici2006}. For a set of $n$ positive semidefinite matrices $\{A_k\}_{k=1}^n$, it holds that
\begin{equation}
\norm{\sum_{k=1}^n A_k}_\infty \leq \norm{\Gamma}_\infty\ ,
\end{equation}
where $\Gamma$ is an $n\times n$ matrix with elements $\Gamma_{i,j}\coloneqq\norm{\sqrt{A_i}\sqrt{A_j}}_\infty$. Applying this to the quantum \ac{ASP} we get
\begin{align}\label{step4}
\mathcal{P}^Q_{n,d,D} =\frac{1}{n d^n}\sum_x \norm{\sum_y M_{x_y|y}}_\infty\leq \frac{1}{n d^n}\sum_x \norm{\Gamma^{(x)}}_\infty,
\end{align}
where $\Gamma^{(x)}_{y,z}=\norm{\sqrt{M_{x_y|y}}\sqrt{M_{x_z|z}}}_\infty$, for $y,z\in [n]$. Let us now fix the computational basis $\{ \ket{y} \}_{y=1}^n$ on $\mathbb{C}^n$, and separate the diagonal and off-diagonal parts of matrices $\Gamma^{(x)}$, 
\begin{equation}
\Gamma^{(x)}=\sum_{y} \norm{M_{x_y|y}}_\infty\ketbra{y}{y}+\sum_{z\neq y} \norm{\sqrt{M_{x_y|y}}\sqrt{M_{x_z|z}}}_\infty\ketbra{y}{z},
\end{equation}
where, as before, $\sum_{y \neq z}$ is the summation over $y\in [n]$ and $z\in [n]$ such that $y \neq z$. Applying the triangle inequality gives
\begin{equation}
\norm{\Gamma^{(x)}}_\infty \leq \norm{\sum_{y} \norm{M_{x_y|y}}_\infty\ketbra{y}{y}}_\infty +\norm{\sum_{y\neq z}\norm{\sqrt{M_{x_y|y}}\sqrt{M_{x_z|z}}}_\infty\ketbra{y}{z}}_\infty.
\end{equation}
The completeness condition of \acp{POVM} implies $\norm{M_{b|y}}_\infty\leq 1$ for all $b\in[d\,], y\in [n]$, which when applied to the first term above bounds it by $1$. 
Substituting the above into the right-hand-side of Eq.~\eqref{step4}, we arrive at
\begin{equation}
\mathcal{P}^Q_{n,d,D} \leq \frac{1}{n}+\frac{1}{n d^n}\sum_x\norm{\sum_{y\neq z}\norm{\sqrt{M_{x_y|y}}\sqrt{M_{x_z|z}}}_\infty\ketbra{y}{z}}_\infty.
\end{equation}
Notice now that the operators $\sum_{y\neq z}\norm{\sqrt{M_{x_y|y}}\sqrt{M_{x_z|z}}}_\infty\ketbra{y}{z}$ are Hermitian and trace-zero. Hence we can apply Lemma~\ref{Lemma} to bound the operator norm by the Frobenius norm. This gives
\begin{equation}
\begin{split}
\mathcal{P}^Q_{n,d,D} & \left. \leq \frac{1}{n}+\frac{\sqrt{n-1}}{n\sqrt{n}d^n}\sum_x\norm{\sum_{y\neq z}\norm{\sqrt{M_{x_y|y}}\sqrt{M_{x_z|z}}}_\infty\ketbra{y}{z}}_F \right. \\
& \left. =\frac{1}{n}+\frac{\sqrt{n-1}}{n\sqrt{n}d^n}\sum_x \sqrt{\sum_{y\neq z} \norm{\sqrt{M_{x_y|y}}\sqrt{M_{x_z|z}}}_\infty^2}.
\right. 
\end{split}
\end{equation}
Bounding the operator norm by the Frobenius norm again, but this time without the trace-zero condition, and therefore without the rank pre-factor, and then using the concavity of the square-root function as well as the completeness condition, we get
\begin{align}
\begin{split}
\mathcal{P}^Q_{n,d,D} & \leq \frac{1}{n}+\frac{\sqrt{n-1}}{n\sqrt{n}d^n}\sum_x \sqrt{\sum_{y\neq z} \norm{\sqrt{M_{x_y|y}}\sqrt{M_{x_z|z}}}_F^2} \\
& = \frac{1}{n}+\frac{\sqrt{n-1}}{n\sqrt{n} d^n} \sum_x \sqrt{\sum_{y\neq z} \Tr{M_{x_y|y}M_{x_z|z}}} \\
& \leq \frac{1}{n}+\frac{\sqrt{n-1}}{n\sqrt{n}\sqrt{d^n}}\sqrt{\sum_{y\neq z} \sum_{x\setminus\{x_y,x_z\}}\sum_{x_y}\sum_{x_z} \Tr{M_{x_y|y}M_{x_z|z}}} = \frac{1}{n}+(n-1)\frac{\sqrt{D}}{nd}\ ,
\end{split}
\end{align}
which is the final result.
\end{proof}

Combining the bounds in Result~\ref{Res1} and Result~\ref{Res2}, the our final bound on the quantum ASP becomes simply the smallest of the two, which is summarized by the following Corollary. 
\begin{corollary}\label{Cor1}
The average success probability of the quantum random access code, in the setting of $n$-element data set with alphabet size $d$ and message dimension $D$, is bounded as
\begin{equation}\label{res3}
\mathcal{P}^Q_{n,d,D} \leq  \min\left\{\frac{1}{d}+\frac{D-1}{\sqrt{ndD}},  \frac{1}{n}\left(1+(n-1)\frac{\sqrt{D}}{d}\right)\right\}.
\end{equation}
In particular, for the case $D=d=n$, the two expressions are identical. 
For the case $D=d$ and $n \geq d$, we have that
\begin{equation}
\mathcal{P}^Q_{n,d,d} \le \frac1d\left( 1 + \frac{d-1}{\sqrt{n}} \right),
\end{equation}
which corresponds to the first bound in Eq.~\eqref{res3}, and for $D=d$ and $n \leq d$, we have that
\begin{equation}\label{eq:n_le_d}
\mathcal{P}^Q_{n,d,d} \le \frac1n\left( 1 + \frac{n-1}{\sqrt{d}} \right),
\end{equation}
which corresponds to the second bound in Eq.~\eqref{res3}.
\end{corollary}
When applied to the special cases of $(n,2,2)$ and $(2,d,d)$, the bound in Corollary~\ref{Cor1} reduces to the previously known bounds in Eq.~\eqref{d2} and Eq.~\eqref{n2}. The bound in Corollary~\ref{Cor1} and the other known generic bound, given in Eq.~\eqref{julio}, admit no strict hierarchy. That is, there exist tuples $(n,d,D)$ for which one of the bounds performs better than the other. For instance, when $D=d$ and $n\geq d$, the two bounds are identical, while for $n\leq d\leq D$, Corollary~\ref{Cor1} provides a tighter bound. This is also the case when $D<d$. On the other hand, for $D>d$ and a sufficiently large $n$, Eq.~\eqref{julio} is a tighter bound. In general, the best known bound can be easily determined for each particular case.

It is worth mentioning that Results~\ref{Res1} and \ref{Res2} can be extended to explicitly account for a noisy communication channel between Alice and Bob.
If we denote by $\Lambda$ the quantum channel used by the parties for communication, then a more accurate formula for the optimal \ac{ASP} becomes
\begin{equation}\label{ASP_channel}
\tilde{\mathcal{P}}^Q_{n,d,D} \coloneqq \max_{\rho_x,M_{b|y}} \frac{1}{nd^n}\sum_{x\in[d\,]^n}\sum_{y=1}^n \Tr{\Lambda(\rho_x) M_{x_y|y}}.
\end{equation}
Clearly, since $\Lambda(\rho_x)$ are again quantum states, the bounds from Corollary~\ref{Cor1} also apply for $\tilde{\mathcal{P}}^Q_{n,d,D}$ in Eq.~\eqref{ASP_channel}.
However, as recently shown in Ref.~\cite{Silva2023}, noisy communication channel can make the maximally attainable \ac{ASP} significantly lower. 

The simplest type of noisy communication channel to consider is the depolarizing channel, whose action on an operator $X$ can be described as $\Lambda(X) = (1-\eta)X + \eta\frac{\Tr{X} \openone}{D}$, where $\eta\in[0,1]$ is the depolarizing parameter. 
Plugging this formula into Eq.~\eqref{ASP_channel} lets us deduce straightforwardly that in this case $\tilde{\mathcal{P}}^Q_{n,d,D} = (1-\eta)\mathcal{P}^Q_{n,d,D}+\frac{\eta}{d}$, where $\mathcal{P}^Q_{n,d,D}$ is the optimal \ac{ASP} from Eq.~\eqref{ASP}.
Other types of noisy channels, e.g., dephasing channel, are less straightforward to account for, and, in our opinion, adapting Results~\ref{Res1} and \ref{Res2} to such cases deserves a separate study. 
However, we can lay out a general strategy that can be followed.
Since in our proofs we eliminate the states from the optimization and only optimize over the \ac{POVM} effects $M_{x_y\vert y}$, it is meaningful to move to the Heisenberg picture and consider the dual of the noise channel. That is, use the defining relation for the dual channel
\begin{equation}
\Tr{\Lambda(\rho_x) M_{x_y|y}} = \Tr{ \rho_x \Lambda^\dagger (M_{x_y|y} )},
\end{equation}
and optimize over \ac{POVM} elements under the fixed dual channel $\Lambda^\dagger$.
In this way, one can potentially use tighter bounds in the proofs of Results~\ref{Res1} and \ref{Res2}, depending on the nature of the map $\Lambda$.
As an example, in Eq.~\eqref{step2_2} we use the inequality $\Tr{M^2_{x_y\vert y}}\leq \Tr{M_{x_y\vert y}}$, which is only tight for projectors. The image of the dual of noisy channels are not projections in most cases. As such, given an explicit description of the channel, the above inequality can likely be tightened.

\section{Numerical case studies}

\begin{figure}[t!]
    \centering
    \includegraphics[width=.6\textwidth]{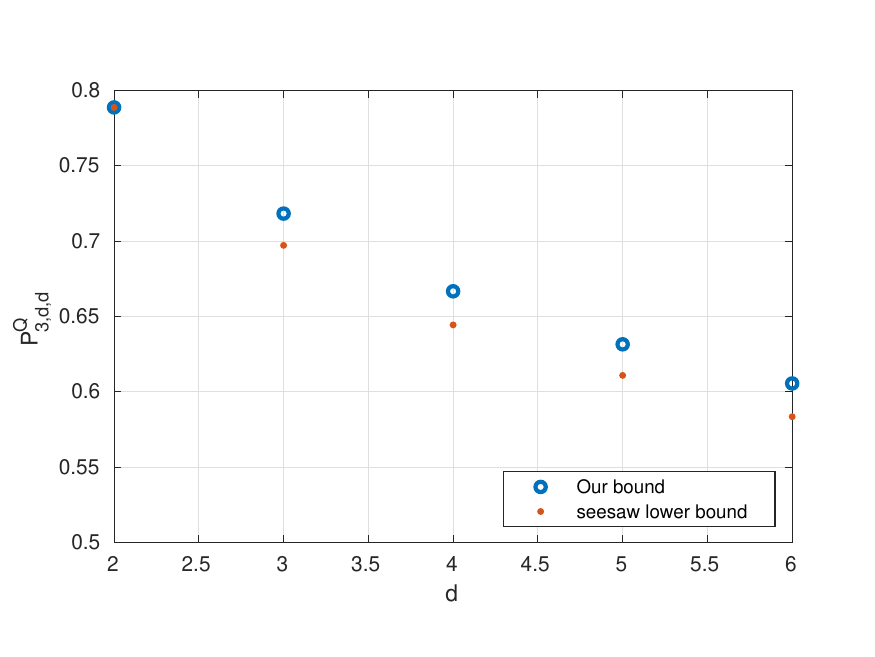}
    \caption{The upper bound from Corollary~\ref{Cor1} and a seesaw lower bound on \ac{ASP} for $n=3$ and $d=D$ with respect to the dimension $D$ of the quantum system.}\label{fig:3dd}
\end{figure}

In this section, we present numerical bounds on the quantum \ac{ASP} for some families of \acp{RAC} in order to benchmark our analytic bounds. We present both lower bounds to see the possible gap between our analytic results and the actual optimal \ac{ASP}, and upper bounds to prove that in some cases our bounds can be improved. Our codes generating the numerical data (apart from those taken from Ref.~\cite{Pauwels2022}) can be found in our open-access repository \cite{Miklin_github}.

We use the \emph{seesaw} \ac{SDP} algorithms to derive lower bounds on the quantum \ac{ASP}. Note that the \ac{ASP} \eqref{ASP} is linear in both $\rho_x$ and in $M_{b|y}$. Consequently, for a fixed \ac{POVM} operators $M_{b|y}$, optimizing the \ac{ASP} over the set of $\rho_x$ is an \ac{SDP}. Similarly, for a fixed set of $\rho_x$, optimizing the \ac{ASP} over $M_{b|y}$ is again an \ac{SDP}. The seesaw algorithm starts with a randomly selected set of $\rho_x$, and finds the optimal $M_{b|y}$ for that set as the solution of an \ac{SDP}. Then $M_{b|y}$ found in the first iteration are fixed, and the optimal $\rho_x$ for this set of \acp{POVM} are found via an \ac{SDP}. This process is iterated until the value of the \ac{ASP} converges to a locally optimal value up to some error threshold. This method is not guaranteed to find a global maximum, but every set of $\rho_x$ and $M_{b|y}$ found by this algorithm provide a valid lower bound on the \ac{ASP}.

The performance of the seesaw algorithm is highly dependent on the initialization, i.e., on the randomly selected $\rho_x$.
To obtain the lower bounds in this paper, in our numerical calculations we sample random pure states $\rho_x$ uniformly (with respect to the Haar measure) and repeat the algorithm a fixed number of times. Since there is no guarantee that the produced estimate is unbiased, one cannot determine a sufficient number of random initialization of the seesaw algorithm.
Nevertheless, by observing the distribution of the local optima produced by the seesaw algorithm for a number of random initializations, one can judge about optimality of the produced lower bounds. 

Since the case of $n=3$ and $d=D$ has been the focus of study for some time \cite{Aguilar2018} and this family is unsolved apart from $d=2$, in Fig.~\ref{fig:3dd} we compare lower bounds from seesaw techniques with our analytic bound from Corollary~\ref{Cor1}, which correspond to Eq.~\eqref{eq:n_le_d} because we have $n \le d$, for values $d \in \{3,4,5,6\}$. We also include the point for $d=2$, for which the seesaw algorithm finds the known optimal \ac{ASP}. This plot demonstrates that our bound is possibly not tight for these values of $d$, but the gap between the lower bound from the seesaw algorithm and our upper bound is seemingly not too large. That is, our upper bounds provide a reasonable approximation of the optimal quantum \ac{ASP} from above.

The smallest unsolved case of quantum \ac{RAC} is $n=d=D=3$, for which the best lower bound, both analytically and numerically, is given by performing measurements in three \acp{MUB}, and choosing the corresponding optimal states \cite{Aguilar2018,Farkas2017}. This gives the lower bound $\mathcal{P}^Q_{3,3,3} \gtrsim 0.6971$, and our analytic upper bound is $\mathcal{P}^Q_{3,3,3} \leq \frac12\left( 1 + \frac{2}{\sqrt{3}} \right) \approx 0.7182$, according to Eq.~\eqref{eq:n_le_d}. In order to test the tightness of our analytic upper bound, we numerically computed upper bounds on $\mathcal{P}^Q_{3,3,3}$ using \ac{SDP} techniques, in particular using the QDimSum package \cite{qdimsum} to implement the techniques described in \cite{Navascues2015,Rosset2019}. Note that this technique, as implemented in \cite{qdimsum}, assumes that the measurements used are projective. While there is no evidence that projective measurements are not optimal for the $n=d=D=3$ case, the upper bounds provided by this technique may be lower than the actual maximum.

Loosely speaking, the method~\cite{Navascues2015} relies on \emph{moment matrices} indexed by monomials of the operators $\rho_x$ and $M_{b|y}$. The more monomials are used for building the moment matrix, the tighter the upper bound, leading to a \emph{hierarchy} of upper bounds. On level 1 of the hierarchy, only order-1 monomials are used, that is, the monomials $\{ \openone \} \cup \{ \rho_x \}_x \cup \{ M_{b|y} \}_{b,y}$. On level 2, all order-2 monomials are used, that is, the monomials
\begin{equation}
\{ \openone \} \cup \{ \rho_x \}_x \cup \{ M_{b|y} \}_{b,y} \cup \{ \rho_x \rho_{x'} \}_{x,x'} \cup \{ M_{b|y} M_{b'|y'} \}_{b,y,b',y'} \cup \{ \rho_x M_{b|y} \}_{x,b,y}.
\end{equation}
One can also define ``intermediate'' levels, for example, the level ``$1 + \rho M$'', using the monomials
\begin{equation}
\{ \openone \} \cup \{ \rho_x \}_x \cup \{ M_{b|y} \}_{b,y} \cup \{ \rho_x M_{b|y} \}_{x,b,y}.
\end{equation}
Using this notation, the upper bounds we found for the various levels of the \ac{SDP} hierarchy for $\mathcal{P}^Q_{3,3,3}$ are given in Table~\ref{table:P333_upper}. 
\begin{table}[h!]
\centering
\begin{tabular}{|l|c|}
\hline
Level & Upper bound on $\mathcal{P}^Q_{3,3,3}$ \\
\hline
1 & 0.7182 \\
\hline
$1 + \rho M$ & 0.6989 \\
\hline
2 & 0.6989 \\
\hline
$2 + MMM$ & 0.6989 \\
\hline
$2 + \rho MM$ & 0.69855 \\
\hline
$2 + MMM + \rho MM$ & 0.69853 \\
\hline
\end{tabular}
\caption{Numerical upper bounds on $\mathcal{P}^Q_{3,3,3}$ using the QDimSum package~\cite{qdimsum} for different levels of the hierarchy.}\label{table:P333_upper}
\end{table}

Note that our analytic bound is recovered, up to the solver precision, for level $1$ of the hierarchy. Then, various consecutive levels yield the same upper bound, which, interestingly, corresponds to the hypothetical case of measuring in three orthonormal bases $\{ \ket{e_j} \}$, $\{ \ket{f_j} \}$ and $\{ \ket{g_j} \}$ on $\mathbb{C}^3$ that form a set of \acp{MUB} with the triple products
\begin{equation}
\braket{e_j}{f_k} \braket{f_k}{g_l} \braket{g_l}{e_j} + \braket{e_j}{g_l} \braket{g_l}{f_k} \braket{f_k}{e_j}
\end{equation}
being uniform, even though such bases are known not to exist~\cite{Farkas2017}.
Note that these triple products are naturally related to Bargmann invariants \cite{Chruciski2004}.
Adding certain order-$3$ monomials, such as $\rho MM$ and $MMM$, makes the upper bound tighter, with our current best numerical upper bound being $\mathcal{P}^Q_{3,3,3} \lesssim 0.69853$, which shows that there is still room for an improvement when deriving analytic upper bounds.

We also compare our analytical bounds to other recently developed \ac{SDP} techniques~\cite{Pauwels2022}. In particular, a special case of the bounds developed in \cite{Pauwels2022} corresponds to upper bounds on $\mathcal{P}^Q_{3,3,D}$. Importantly, these do not assume that the measurements are projective. 
In Figure \ref{fig:33D} we plot the seesaw lower bound, our analytical upper bound, and the upper bound from Ref.~\cite{Pauwels2022} on $\mathcal{P}^Q_{3,3,D}$ for $D \in \{2,3,\ldots, 10\}$. These numerical results show that our bound is not tight for this family of \acp{RAC} apart from the $D=2$ case, where our analytical upper bound is actually tighter than the one obtained in Ref.~\cite{Pauwels2022}.

\begin{figure}[h!]
	\centering
	\includegraphics[width=.6\textwidth]{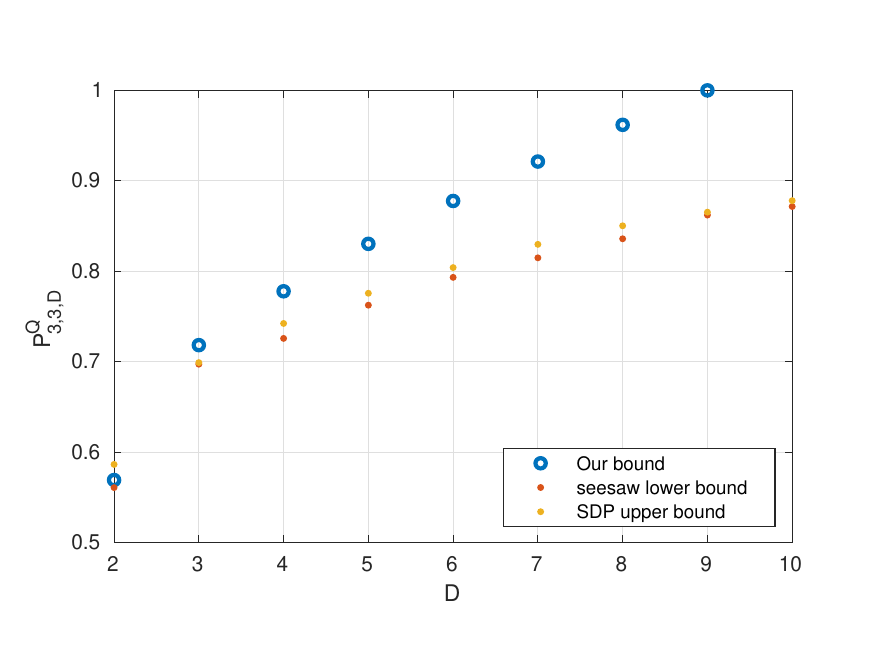}
	\caption{The upper bound from Corollary~\ref{Cor1}, the upper bound from the \ac{SDP} hierarchy of Ref.~\cite{Pauwels2022}, and a seesaw lower bound on \ac{ASP} for $n=3$, $d=3$ with respect to the dimension $D$ of the quantum system.}\label{fig:33D}
\end{figure}

Lastly, we compare our analytical upper bounds with seesaw lower bounds in the $n=3$ case for various different values for $d$ and $D$ in Table \ref{table:P3dD_seesaw}. While our bound is potentially not tight apart from the $d=D=2$ case, it provides a good approximation of the quantum optimal \ac{ASP} for generic values of $d$ and $D$.

\begin{table}[h!]
\centering
\begin{tabular}{|l|c|c|}
\hline
$(d,D)$ & Seesaw lower bound & Our upper bound \\
\hline
(2,2) & 0.78868 & 0.78868 \\ \hline
(2,3) & 0.80794$^\ast$  & 0.91068\\ \hline
(2,4) & 0.90825  & 1\\ \hline
(3,2) & 0.56066  & 0.56904\\ \hline
(3,3) & 0.69715  & 0.71823\\ \hline
(3,4) & 0.72567  & 0.77778\\ \hline
(3,5) & 0.76241  & 0.83024\\ \hline
(4,2) & 0.43697  & 0.45412\\ \hline
(4,3) & 0.47525  & 0.58333\\ \hline
(4,4) & 0.64434  & 0.66667\\ \hline
(4,5) & 0.66331$^\ast$  & 0.70601 \\ \hline
\end{tabular}
\caption{Lower bounds on $\mathcal{P}^Q_{3,d,D}$ for various values of $(d,D)$ from the seesaw method and the upper bound from Corollary~\ref{Cor1}. The presented estimates for the lower bound are the maximal obtained values of \ac{ASP} for $100$ runs of the seesaw algorithm with random initial states. $^\ast$The cases $(2,3)$ and $(4,5)$ appear to be special in a way that the seesaw algorithm often finds other local minima, which is not the case for other cases in this table.}\label{table:P3dD_seesaw}
\end{table}

\section{Conclusions}
In this paper, we derive a universal analytic upper bound on the average success probability of quantum random access codes. We consider the most general case of \acp{RAC} with $n$ independent variables from a $d$-dimensional alphabet encoded into a $D$-dimensional quantum system, for arbitrary $n,d$ and $D$. Our bounds recover known families of upper bounds for the case of $d=D$ and $n=2$, which is known to be tight \cite{Farkas2019}, and the $d=D=2$ case, which is known to be tight for $n=2$ and $3$~\cite{ambainis2009quantum}. In the general case, our bounds are not tight, although numerical evidence suggests that they provide reasonably good approximations of the actual quantum maximum. For the interesting case of $n=d=3$, we show numerically that our upper bounds can be improved, and that the question whether \acp{MUB} measurements are optimal for this case is still unresolved. We believe that our bounds, in combination with that obtained in~\cite{Vicente2019}, will be useful for the analytical treatment of problems where \acp{RAC} are used as a tool for benchmarking quantum information protocols.

Apart from benchmarking, approximate bounds on ASP of RACs are useful for analyzing semi-device-independent quantum key distribution (QKD) protocols \cite{Pawlowski2011}. As discussed in \cite{Pawlowski2011}, the security of some common QKD protocols, such as BB84 \cite{BB84}, cannot be guaranteed if the devices used for quantum communication are not characterized. To alleviate this issue, we can use protocols based on dimension witnesses---such as RACs---which provide security without a full characerization. For the case of RACs, the amount of distillable key only depends on the ASP and not on a full characterization of the devices~\cite{CsK78}.

Note that both of our bounds in Result \ref{Res1} and \ref{Res2} rely on pairwise properties of the measurements. In particular, both of these bounds at some point involve Hilbert-Schmidt inner products of measurement effects, and are, in principle, maximized when these inner products are uniform as both bounds use the concavity of the square root function. This would imply that if \acp{MUB} are indeed optimal for \acp{RAC}, then \emph{any} set of \acp{MUB} should give rise to the optimal \ac{ASP}, as all sets of \acp{MUB} satisfy the above trace uniformity condition. It is known, however, that not all sets of \acp{MUB} give rise to the same \ac{ASP} when $n=3$ \cite{Aguilar2018,Farkas2017}. This implies that either \acp{MUB} are not optimal for \acp{RAC} in the general case, or that new upper bounds relying on more than pairwise properties of measurements are needed to prove optimality.

A potential new avenue for obtaining improved analytic upper bounds is highlighted by the relatively strong performance of the numerical techniques in Ref.~\cite{Pauwels2022}. While these bounds are numerical, they are based on well-established \ac{SDP} hierarchy techniques. Therefore, analysing the dual \ac{SDP} hierarchy induced by that in Ref.~\cite{Pauwels2022} is a promising direction for deriving analytic bounds stronger than those presented in this work.

\begin{acknowledgements}
We thank Marcin Paw\l{}owski for the years of insightful discussions about the subject of this paper. We thank Julio de Vicente for a valuable feedback.
A.T.~is supported by the Knut and Alice Wallenberg Foundation through the Wallenberg Center for Quantum Technology (WACQT) and the Swedish Research Council under Contract No.~2023-03498.
This research was funded by the Deutsche Forschungsgemeinschaft (DFG, German Research Foundation), project number 441423094, and the Fujitsu Germany GmbH as part of the endowed professorship ``Quantum Inspired and Quantum Optimization''.
\end{acknowledgements}

\bibliographystyle{quantum}

\bibliography{references_qrac}

\end{document}